\documentclass[11pt]{article}

\usepackage{amsmath,amsfonts,amssymb,amsthm}
\usepackage[margin=1in]{geometry}
\usepackage[colorlinks]{hyperref}
\usepackage{mathrsfs}
\usepackage{mathtools}
\mathtoolsset{centercolon}
\usepackage[numbers,comma,sort&compress]{natbib}

\newcommand{\C}{{\mathbb{C}}}

\newcommand{\N}{{\mathbb{N}}}

\newcommand{\Z}{{\mathbb{Z}}}

\renewcommand{\d}{\mathrm{d}}

\renewcommand{\Re}{\mathop{\mathrm{Re}}}

\newcommand{\norm}[1]{\|{#1}\|}

\newcommand{\range}[1]{[{#1}]}
\newcommand{\rangez}[1]{[{#1}]_0}

\DeclareMathOperator{\poly}{poly}
\DeclareMathOperator{\diag}{diag}

\newtheorem{theorem}{Theorem}
\newtheorem{lemma}{Lemma}
\newtheorem{corollary}{Corollary}
\newtheorem{problem}{Problem}

\numberwithin{equation}{section}

\newcommand{\eq}[1]{(\ref{eq:#1})}
\renewcommand{\sec}[1]{\hyperref[sec:#1]{Section~\ref*{sec:#1}}}
\newcommand{\app}[1]{\hyperref[app:#1]{Appendix~\ref*{app:#1}}}
\newcommand{\thm}[1]{\hyperref[thm:#1]{Theorem~\ref*{thm:#1}}}
\newcommand{\lem}[1]{\hyperref[lem:#1]{Lemma~\ref*{lem:#1}}}
\newcommand{\cor}[1]{\hyperref[cor:#1]{Corollary~\ref*{cor:#1}}}
\newcommand{\prb}[1]{\hyperref[prb:#1]{Problem~\ref*{prb:#1}}}

\begin{document}


\title{Quantum spectral methods for differential equations}

\author{Andrew M.\ Childs$^{1,2,3}$ \qquad Jin-Peng Liu$^{1,4}$\\[2pt]
\small $^{1}$Joint Center for Quantum Information and Computer Science, \\
\small $^{2}$Department of Computer Science,
\small $^{3}$Institute for Advanced Computer Studies, and \\
\small $^{4}$Applied Mathematics, Statistics, and Scientific Computation, \\
\small University of Maryland
}

\date{}
\maketitle


\begin{abstract}
Recently developed quantum algorithms address computational challenges in numerical analysis by performing linear algebra in Hilbert space. Such algorithms can produce a quantum state proportional to the solution of a $d$-dimensional system of linear equations or linear differential equations with complexity $\poly(\log d)$. While several of these algorithms approximate the solution to within $\epsilon$ with complexity $\poly(\log(1/\epsilon))$, no such algorithm was previously known for differential equations with time-dependent coefficients. Here we develop a quantum algorithm for linear ordinary differential equations based on so-called spectral methods, an alternative to finite difference methods that approximates the solution globally. Using this approach, we give a quantum algorithm for time-dependent initial and boundary value problems with complexity $\poly(\log d, \log(1/\epsilon))$.
\end{abstract}

\section{Introduction}
\label{sec:introduction}

Differential equations have extensive applications throughout mathematics, science, and engineering. Numerical methods for differential equations have been widely studied (see for example \cite{Kre98}), giving fast algorithms for solving them using classical computers.

Recent work has developed quantum algorithms with the potential to extract information about solutions of systems of differential equations even faster than is possible classically. This body of work grew from the quantum linear systems algorithm (QLSA) \cite{HHL09}, which produces a quantum state proportional to the solution of a sparse system of $d$ linear equations in time $\poly(\log d)$. Subsequent work improved the performance of that algorithm \cite{Amb12,CKS15} and applied it to develop similar quantum algorithms for differential equations.

To achieve this improvement, quantum algorithms must operate under different assumptions than those made for algorithms in classical numerical analysis. To represent the output using $\poly(\log d)$ qubits, the output is produced as a quantum state, not as an explicit description of a vector. Furthermore, to facilitate applying Hamiltonian simulation, these quantum algorithms use implicit access to the system of equations (say, through a matrix specified by a sparse Hamiltonian oracle and the ability to prepare quantum states encoding certain vectors). While these assumptions restrict the types of equations that can be addressed and the types of information that can be extracted from the final state, they nevertheless appear capable of producing useful information that cannot be efficiently computed classically.

In this paper, we focus on systems of first-order linear ordinary differential equations (ODEs). Such equations can be written in the form
\begin{equation}
\frac{\d{x(t)}}{\d{t}}=A(t)x(t)+f(t)
\label{eq:ode}
\end{equation}
where $t \in [0,T]$ for some $T>0$, the solution $x(t)\in\C^d$ is a $d$-dimensional vector, and the system is determined by a time-dependent matrix $A(t)\in\C^{d\times d}$ and a time-dependent inhomogeneity $f(t)\in\C^d$. Provided $A(t)$ and $f(t)$ are continuous functions of $t$, the initial value problem (i.e., the problem of determining $x(t)$ for a given initial condition $x(0)$) has a unique solution \cite{Apo63}.

The Hamiltonian simulation problem is simply the special case of the quantum ODE problem where $A$ is antihermitian and $f$ is zero. A substantial body of work has developed fast quantum algorithms for that special case  \cite{BAC07,PQSV11,BCC13,BCC15,BCK15,BHMT02,BN16,LC16,LC17,NB16}. Hamiltonian simulation underlies the QLSA \cite{HHL09,CKS15} which in turn gives algorithms for more general differential equations.

Berry presented the first efficient quantum algorithm for general linear ODEs \cite{Ber14}. His algorithm represents the system of differential equations as a system of linear equations using a linear multistep method and solves that system using the QLSA. This approach achieves complexity logarithmic in the dimension $d$ and, by using a high-order integrator, close to quadratic in the evolution time $T$. While this method could in principle be applied to handle time-dependent equations, the analysis of \cite{Ber14} only explicitly considers the time-independent case for simplicity.

Since it uses a finite difference approximation, the complexity of Berry's algorithm as a function of the solution error $\epsilon$ is $\poly(1/\epsilon)$ \cite{Ber14}. Reference \cite{BCOW17} improved this to $\poly(\log(1/\epsilon))$ by using a high-precision QLSA based on linear combinations of unitaries \cite{CKS15} to solve a linear system that encodes a truncated Taylor series. However, this approach assumes that $A(t)$ and $f(t)$ are time-independent so that the solution of the ODE can be written as an explicit series, and it is unclear how to generalize the algorithm to time-dependent ODEs.

While we focus here on extending the above line of work, several other approaches have been proposed for addressing differential equations with quantum computers. Reference \cite{LO08} used a quantum version of the Euler method to handle nonlinear ODEs with polynomial nonlinearities. This algorithm has complexity logarithmic in the dimension but exponential in the evolution time (as is inevitable for general nonlinear ODEs). Other work has developed quantum algorithms for partial differential equations (PDEs). Reference \cite{CJS13} described a quantum algorithm that applies the QLSA to implement a finite element method for Maxwell's equations. Reference \cite{CJO17} applied Hamiltonian simulation to a finite difference approximation of the wave equation. Most recently, reference \cite{AKWL18} presented a continuous-variable quantum algorithm for initial value problems with non-homogeneous linear PDEs.

Most of the aforementioned algorithms use a local approximation: they discretize the differential equations into small time intervals to obtain a system of linear equations or linear differential equations that can be solved by the QLSA or Hamiltonian simulation. For example, the central difference scheme approximates the time derivative at the point $x(t)$ as
\begin{equation}
\frac{\d{x(t)}}{\d{t}}=\frac{x(t+h)-x(t-h)}{2h}+O(h^2).
\end{equation}
High-order finite difference or finite element methods can reduce the error to $O(h^k)$, where $k-1$ is the order of the approximation. However, when solving an equation over the interval $[0,T]$, the number of iterations is $T/h = \Theta(\epsilon^{-1/k})$ for fixed $k$, giving a total complexity that is $\poly(1/\epsilon)$ even using high-precision methods for the QLSA or Hamiltonian simulation.

For ODEs with special structure, some prior results already show how to avoid a local approximation and thereby achieve complexity $\poly(\log(1/\epsilon))$. When $A(t)$ is anti-Hermitian and $f(t)=0$, we can directly apply Hamiltonian simulation \cite{BCK15}; if $A$ and $f$ are time-independent, then \cite{BCOW17} uses a Taylor series to achieve complexity $\poly(\log({1}/{\epsilon}))$. However, the case of general time-dependent linear ODEs had remained elusive.

In this paper, we use a nonlocal representation of the solution of a system of differential equations to give a new quantum algorithm with complexity $\poly(\log(1/\epsilon))$ even for time-dependent equations. Time-dependent linear differential equations describe a wide variety of systems in science and engineering. Examples include the wave equation and the Stokes equation (i.e., creeping flow) in fluid dynamics \cite{KK13}, the heat equation and the Boltzmann equation in thermodynamics \cite{Tha11,Har07}, the Poisson equation and Maxwell's equations in electromagnetism \cite{Jac12,Thi04}, and of course Schr{\"o}dinger's equation in quantum mechanics. Moreover, some nonlinear differential equations can be studied by linearizing them to produce time-dependent linear equations (e.g., the linearized advection equation in fluid dynamics \cite{Bir02}).

We focus our discussion on first-order linear ODEs. Higher-order ODEs can be transformed into first-order ODEs by standard methods. Also, by discretizing space, PDEs with both time and space dependence can be regarded as sparse linear systems of time-dependent ODEs. Thus we focus on an equation of the form \eq{ode} with initial condition
\begin{equation}
x(0)=\gamma
\end{equation}
for some specified $\gamma\in\C^d$. We assume that $A(t)$ is $s$-sparse (i.e., has at most $s$ nonzero entries in any row or column) for any $t\in[0,T]$. Furthermore, we assume that $A(t)$, $f(t)$, and $\gamma$ are provided by black-box subroutines (which serve as abstractions of efficient computations). In particular, following essentially the same model as in \cite{BCOW17} (see also Section 1.1 of \cite{CKS15}), suppose we have an oracle $O_A(t)$ that, for any $t \in [0,T]$ and any given row or column specified as input, computes the locations and values of the nonzero entries of $A(t)$ in that row or column. We also assume oracles $O_x$ and $O_f(t)$ that, for any $t \in [0,T]$, prepare normalized states $|\gamma\rangle$ and $|f(t)\rangle$ proportional to $\gamma$ and $f(t)$, and that also compute $\|\gamma\|$ and $\|f(t)\|$, respectively. Given such a description of the instance, the goal is to produce a quantum state $\epsilon$-close to $|x(T)\rangle$ (a normalized quantum state proportional to $x(T)$).

As mentioned above, our main contribution is to implement a method that uses a global approximation of the solution. We do this by developing a quantum version of so-called \emph{spectral methods}, a technique from classical numerical analysis that (approximately) represents the components of the solution $x(t)_i \approx \sum_j c_{ij} \phi_j(t)$ as linear combinations of basis functions $\phi_j(t)$ expressing the time dependence. Specifically, we implement a Chebyshev pseudospectral method \cite{BH02,Hos06} using the QLSA. This approach approximates the solution by a truncated Chebyshev series with undetermined coefficients and solves for those coefficients using a linear system that interpolates the differential equations. According to the convergence theory of spectral methods, the solution error decreases exponentially provided the solution is sufficiently smooth \cite{Ghe07,STW11}. We use the LCU-based QLSA to solve this linear system with high precision \cite{CKS15}. To analyze the algorithm, we upper bound the solution error and condition number of the linear system and lower bound the success probability of the final measurement. Overall, we show that the total complexity of this approach is $\poly(\log(1/\epsilon))$ for general time-dependent ODEs.

In addition to initial value problems, our approach can also address boundary value problems (BVPs). Given an oracle for preparing a state $\alpha|x(0)\rangle+\beta|x(T)\rangle$ expressing a general boundary condition, the goal of the quantum BVP is to produce a quantum state $\epsilon$-close to $|x(t)\rangle$ (a normalized state proportional to $x(t)$) for any desired $t \in [0,T]$. We also give a quantum algorithm for this problem with complexity $\poly(\log(1/\epsilon))$.

The remainder of this paper is organized as follows. \sec{spectral} introduces the spectral method and \sec{linear_system} shows how to encode it into a quantum linear system. Then \sec{error} analyzes the exponential decrease of the solution error, \sec{condition_number} bounds the condition number of the linear system, \sec{success_prob} lower bounds the success probability of the final measurement, and \sec{state preparation} describes how to prepare the initial quantum state. We combine these bounds in \sec{main} to establish the main result. We then extend the analysis for initial value problems to boundary value problems in \sec{BVP}. Finally, we conclude in \sec{discussion} with a discussion of the results and some open problems.

\section{Spectral method}
\label{sec:spectral}

Spectral methods provide a way of solving differential equations using global approximations \cite{Ghe07,STW11}. The main idea of the approach is as follows. First, express an approximate solution as a linear combination of certain basis functions with undetermined coefficients. Second, construct a system of linear equations that such an approximate solution should satisfy. Finally, solve the linear system to determine the coefficients of the linear combination.

Spectral methods offer a flexible approach that can be adapted to different settings by careful choice of the basis functions and the linear system. A Fourier series provides an appropriate basis for periodic problems, whereas Chebyshev polynomials can be applied more generally. The linear system can be specified using Gaussian quadrature (giving a \emph{spectral element method} or \emph{Tau method}), or one can simply interpolate the differential equations using quadrature nodes (giving a \emph{pseudo-spectral method}) \cite{STW11}. Since general linear ODEs are non-periodic, and interpolation facilitates constructing a straightforward linear system, we develop a quantum algorithm based on the Chebyshev pseudo-spectral method \cite{BH02,Hos06}.

In this approach, we consider a truncated Chebyshev approximation $x(t)$ of the exact solution $\hat x(t)$, namely
\begin{equation}
x_i(t)=\sum_{k=0}^nc_{i,k}T_k(t),\quad i\in\rangez{d}:=\{0,1,\ldots,d-1\}
\label{eq:cheby_expand}
\end{equation}
for any $n \in \Z^+$. (See \app{Chebyshev} for the definition of $T_k(t)$ and a discussion of its properties.) The coefficients $c_{i,k} \in \C$ for all $i\in\rangez{d}$ and $k\in\rangez{n+1}$ are determined by demanding that $x(t)$ satisfies the ODE and initial conditions at a set of \emph{interpolation nodes} $\{t_l\}_{l=0}^n$ (with $1=t_0>t_1>\cdots>t_n=-1$), where $x(t_0)$ and $x(t_n)$ are the initial and final states, respectively. In other words, we require
\begin{equation}
\frac{\d{x(t_l)}}{\d{t}}=A(t_l)x(t_l)+f(t_l), \quad \forall \, l\in\range{n+1},~ t\in[-1,1],
\label{eq:cheby_ode}
\end{equation}
and
\begin{equation}
x_i(t_0)=\gamma_i, \quad i\in\rangez{d}.
\label{eq:ode_init}
\end{equation}
We would like to be able to increase the accuracy of the approximation by increasing $n$, so that
\begin{equation}
\|\hat{x}(t)-x(t)\|\rightarrow0 \quad \text{as} \quad n\rightarrow\infty.
\end{equation}

There are many possible choices for the interpolation nodes. Here we use the Chebyshev-Gauss-Lobatto quadrature nodes, $t_l=\cos\frac{l\pi}{n}$ for $l \in \rangez{n+1}$, since these nodes achieve the highest convergence rate among all schemes with the same number of nodes \cite{DCS89,Kre98}. These nodes also have the convenient property that $T_k(t_l)=\cos\frac{kl\pi}{n}$, making it easy to compute the values $x_i(t_l)$.

To evaluate the condition \eq{cheby_ode}, it is convenient to define coefficients $c'_{i,k}$ for $i\in\rangez{d}$ and $k\in\rangez{n+1}$ such that
\begin{equation}
\frac{\d{x_i(t)}}{\d{t}}=\sum_{k=0}^n c'_{i,k}T_k(t).
\label{eq:cheby_deriv}
\end{equation}
We can use the differential property of Chebyshev polynomials,
\begin{align}
  2T_k(t)=\frac{T'_{k+1}(t)}{k+1}-\frac{T'_{k-1}(t)}{k-1},
\end{align}
to determine the transformation between $c_{i,k}$ and $c'_{i,k}$. As detailed in \app{Chebyshev}, we have
\begin{equation}
c'_{i,k}=\sum_{j=0}^n[D_n]_{kj}c_{i,j},\quad i\in\rangez{d},~k\in\rangez{n+1},
\label{eq:cheby_c'}
\end{equation}
where $D_n$ is the $(n+1) \times (n+1)$ upper triangular matrix with nonzero entries
\begin{equation}
[D_n]_{kj}=\frac{2j}{\sigma_k},\qquad \text{$k+j$ odd},~ j>k,
\label{eq:cheby_Dn}
\end{equation}
where
\begin{equation}
\sigma_k := \begin{cases}
2 & k=0\\
1 & k\in\range{n} := \{1,2,\ldots,n\}.
\end{cases}
\end{equation}

Using this expression in \eq{cheby_ode}, \eq{cheby_c'}, and \eq{cheby_Dn}, we obtain the following linear equations:
\begin{equation}
\label{eq:linear_eq}
\sum_{k=0}^nT_k(t_l)c'_{i,k}=\sum_{j=0}^{d-1}A_{ij}(t_l)\sum_{k=0}^nT_k(t_l)c_{j,k}+f(t_l)_i,\quad i\in\rangez{d},~l\in\rangez{n+1}.
\end{equation}
We also demand that the Chebyshev series satisfies the initial condition $x_i(1)=\gamma_i$ for all $i\in\rangez{d}$. This system of linear equations gives a global approximation of the underlying system of differential equations. Instead of locally approximating the ODE at discretized times, these linear equations use the behavior of the differential equations at the $n+1$ times $\{t_l\}_{l=0}^n$ to capture their behavior over the entire interval $[-1,1]$.

Our algorithm solves this linear system using the high-precision QLSA \cite{CKS15}. Given an encoding of the Chebyshev coefficients $c_{ik}$, we can obtain the approximate solution $x(t)$ as a suitable linear combination of the $c_{ik}$, a computation that can also be captured within a linear system. The resulting approximate solution $x(t)$ is close to the exact solution $\hat{x}(t)$:

\begin{lemma}[Lemma 19 of \cite{Ghe07}]\label{lem:spectral_approx}
Let $\hat{x}(t)\in C^{r+1}(-1,1)$ be the solution of the differential equations \eq{ode} and let $x(t)$ satisfy \eq{cheby_ode} and \eq{ode_init} for $\{t_l=\cos\frac{l\pi}{n}\}_{l=0}^n$. Then there is a constant $C$, independent of $n$, such that
\begin{equation}
\max_{t\in[-1,1]}\|\hat{x}(t)-x(t)\|\le C \max_{t\in[-1,1]}\frac{\|\hat{x}^{(n+1)}(t)\|}{n^{r-2}}.
\end{equation}
\end{lemma}

This shows that the convergence behavior of the spectral method is related to the smoothness of the solution. For a solution in $C^{r+1}$, the spectral method approximates the solution with $n=\poly({1}/{\epsilon})$. Furthermore, if the solution is smoother, we have an even tighter bound:

\begin{lemma}[Eq.~(1.8.28) of \cite{STW11}]\label{lem:smooth_spectral_approx}
Let $\hat{x}(t)\in C^{\infty}(-1,1)$ be the solution of the differential equations \eq{ode} and let $x(t)$ satisfy \eq{cheby_ode} and \eq{ode_init} for $\{t_l=\cos\frac{l\pi}{n}\}_{l=0}^n$. Then
\begin{equation}
\max_{t\in[-1,1]}\|\hat{x}(t)-x(t)\|\le \sqrt{\frac{2}{\pi}} \max_{t\in[-1,1]}\|\hat{x}^{(n+1)}(t)\|\left(\frac{e}{2n}\right)^n.
\end{equation}
\end{lemma}

For simplicity, we replace the value $\sqrt{{2}/{\pi}}$ by the upper bound of $1$ in the following analysis.

This result implies that if the solution is in $C^{\infty}$, the spectral method approximates the solution to within $\epsilon$ using only $n=\poly(\log(1/\epsilon))$ terms in the Chebyshev series. Consequently, this approach gives a quantum algorithm with complexity $\poly(\log(1/\epsilon))$.

\section{Linear system}
\label{sec:linear_system}

In this section we construct a linear system that encodes the solution of a system of differential equations via the Chebyshev pseudospectral method introduced in \sec{spectral}. We consider a system of linear, first-order, time-dependent ordinary differential equations, and focus on the following initial value problem:

\begin{problem}\label{prb:ode}
In the \emph{quantum ODE problem}, we are given a system of equations
\begin{equation}
\frac{\d{x(t)}}{\d{t}}=A(t)x(t)+f(t)
\end{equation}
where $x(t) \in\C^d$, $A(t) \in \C^{d\times d}$ is $s$-sparse, and $f(t) \in \C^d$ for all $t\in[0,T]$. We assume that $A_{ij},f_i \in C^\infty(0,T)$ for all $i,j \in \range{d}$. We are also given an initial condition $x(0) = \gamma \in \C^d$. Given oracles that compute the locations and values of nonzero entries of $A(t)$ for any $t$, and that prepare normalized states $|\gamma\rangle$ proportional to $\gamma$ and $|f(t)\rangle$ proportional to $f(t)$ for any $t \in [0,T]$, the goal is to output a quantum state $|x(T)\rangle$ that is proportional to $x(T)$.
\end{problem}

Without loss of generality, we rescale the interval $[0,T]$ onto $[-1,1]$ by the linear map $t \mapsto 1-2t/T$. Under this rescaling, we have $\frac{\d}{\d{t}} \mapsto -\frac{T}{2}\frac{\d}{\d{t}}$, so $A \mapsto -\frac{T}{2}A$, which can dramatically increase the spectral norm. To avoid this phenomenon, we divide the interval $[0,T]$ into subintervals $[0,\Gamma_1],[\Gamma_1,\Gamma_2],\ldots,[\Gamma_{m-1},T]$ with $\Gamma_0:=0, \Gamma_m:=T$.
Each subinterval $[\Gamma_h,\Gamma_{h+1}]$ for $h\in\rangez{m}$ is then rescaled onto $[-1,1]$ with the linear map $K_h \colon [\Gamma_h,\Gamma_{h+1}] \to [-1,1]$ defined by
\begin{equation}
  \begin{aligned}
    K_h \colon t \mapsto 1-\frac{2(t-\Gamma_h)}{\Gamma_{h+1}-\Gamma_h},
  \end{aligned}
\label{eq:rescaling_map}
\end{equation}
which satisfies $K_h(\Gamma_h)=1$ and $K_h(\Gamma_{h+1})=-1$. To solve the overall initial value problem, we simply solve the differential equations for each successive interval (as encoded into a single system of linear equations).

Now let $\tau_h:=|\Gamma_{h+1}-\Gamma_h|$ and define
\begin{align}
  A_h(t) &:= -\frac{\tau_h}{2}A(K_h(t)) \label{eq:Ah} \\
  \label{eq:rescaling_rate}
  x_h(t) &:= x(K_h(t)) \\
  f_h(t) &:= -\frac{\tau_h}{2}f(K_h(t)).
\end{align}
Then, for each $h\in\rangez{m}$, we have the rescaled differential equations
\begin{equation}
\frac{\d{x_h}}{\d{t}}=A_h(t)x_h(t)+f_h(t)
\end{equation}
for $t\in[-1,1]$ with the initial conditions
\begin{equation}
  x_h(1) = \begin{cases}
  \gamma & h=0 \\
  x_{h-1}(-1) & h\in\range{m}.
  \end{cases}
\end{equation}
By taking
\begin{equation}
\label{eq:tau}
\tau_h \le \frac{2}{\max_{t \in [\Gamma_h,\Gamma_{h+1}]} \norm{A(t)}}
\end{equation}
where $\|{\cdot}\|$ denotes the spectral norm,
we can ensure that $\norm{A_h(t)} \le 1$ for all $t \in [-1,1]$. In particular, it suffices to take
\begin{equation}
\label{eq:tau_max}
  \tau := \max_{h \in \{0,1,\ldots,m-1\}} \tau_h \le \frac{2}{\max_{t \in [0,T]} \norm{A(t)}}.
\end{equation}

Having rescaled the equations to use the domain $[-1,1]$, we now apply the Chebyshev pseudospectral method. Following \sec{spectral}, we substitute the truncated Chebyshev series of $x(t)$ into the differential equations with interpolating nodes $\{t_l=\cos\frac{l\pi}{n} : l \in \range{n}\}$, giving the linear system
\begin{equation}
\frac{\d{x(t_l)}}{\d{t}}=A_h(t_l)x(t_l)+f_h(t_l), \quad h\in\rangez{m}, l\in\range{n+1}
\label{eq:interpolation_diff}
\end{equation}
with initial condition
\begin{equation}
x(t_0)=\gamma.
\label{eq:interpolation_init}
\end{equation}
Note that in the following, terms with $l=0$ refer to this initial condition.

We now describe a linear system
\begin{equation}
L|X\rangle=|B\rangle
\label{eq:linear_system}
\end{equation}
that encodes the Chebyshev pseudospectral approximation and uses it to produce an approximation of the solution at time $T$.

The vector $|X\rangle \in \C^{m+p} \otimes \C^d \otimes \C^{n+1}$ represents the solution in the form
\begin{equation}
|X\rangle=\sum_{h=0}^{m-1}\sum_{i=0}^{d-1}\sum_{l=0}^nc_{i,l}(\Gamma_{h+1})|hil\rangle+\sum_{h=m}^{m+p}\sum_{i=0}^{d-1}\sum_{l=0}^nx_i|hil\rangle
\label{eq:vectorX}
\end{equation}
where $c_{i,l}(\Gamma_{h+1})$ are the Chebyshev series coefficients of $x(\Gamma_{h+1})$ and $x_i := x(\Gamma_m)_i$ is the $i$th component of the final state $x(\Gamma_m)$.

The right-hand-side vector $|B\rangle$ represents the input terms in the form
\begin{equation}
|B\rangle=\sum_{h=0}^{m-1}|h\rangle|B(f_h)\rangle
\label{eq:vectorB}
\end{equation}
where
\begin{equation}
|B(f_h)\rangle=\sum_{i=0}^{d-1}\gamma_i|i0\rangle+\sum_{i=0}^{d-1}\sum_{l=1}^nf_{h}(\cos\tfrac{l\pi}{n})_i|il\rangle,\quad h\in\range{m-1}.
\label{eq:B_f_h}
\end{equation}
Here $\gamma$ is the initial condition and $f_{h}(\cos\frac{l\pi}{n})_i$ is $i$th component of $f_h$ at the interpolation point $t_l=\cos\frac{l\pi}{n}$.

We decompose the matrix $L$ in the form
\begin{equation}
L=\sum_{h=0}^{m-1}|h\rangle\langle h|\otimes(L_1+L_2(A_h))+\sum_{h=1}^{m}|h\rangle\langle h-1|\otimes L_3+\sum_{h=m}^{m+p}|h\rangle\langle h|\otimes L_4+\sum_{h=m+1}^{m+p}|h\rangle\langle h-1|\otimes L_5.
\label{eq:matrixL}
\end{equation}
We now describe each of the matrices $L_i$ for $i \in \range{5}$ in turn.

The matrix $L_1$ is a discrete representation of $\frac{\d{x}}{\d{t}}$, satisfying
\begin{equation}
|h\rangle\langle h| \otimes L_1|X\rangle = \sum_{i=0}^{d-1}\sum_{k=0}^nT_k(t_0)c_{i,k} |hi0\rangle+\sum_{i=0}^{d-1}\sum_{l=1,k,r=0}^nT_k(t_l)[D_n]_{kr}c_{i,r} |hil\rangle
\end{equation}
(recall from \eq{cheby_deriv} and \eq{cheby_c'} that $D_n$ encodes the action of the time derivative on a Chebyshev expansion).
Thus $L_1$ has the form
\begin{align}
L_1
&=\sum_{i=0}^{d-1}\sum_{k=0}^nT_k(t_0)|i0\rangle\langle ik|+\sum_{i=0}^{d-1}\sum_{l=1,k,r=0}^n\cos\frac{kl\pi}{n}[D_n]_{kr}|il\rangle\langle ir| \\
&=I_d\otimes (|0\rangle\langle 0|P_n+\sum_{l=1}^n|l\rangle\langle l|P_nD_n)
\end{align}
where the interpolation matrix is a \emph{discrete cosine transform} matrix:
\begin{equation}
P_n:=\sum_{l,k=0}^n \cos\frac{kl\pi}{n}|l\rangle\langle k|.
\label{eq:dct}
\end{equation}

The matrix $L_2(A_h)$ discretizes $A_h(t)$, i.e.,
\begin{align}
  |h\rangle\langle h| \otimes L_2(A_h)|X\rangle
  =-\sum_{i,j=0}^{d-1}\sum_{l=1,k=0}^n A_h(t_l)_{ij}T_k(t_l)c_{j,k}|hil\rangle.
\end{align}
Thus
\begin{align}
L_2(A_h)
&=-\sum_{i,j=0}^{d-1}\sum_{l=1,k=0}^n A_h(t_l)_{ij} \cos\frac{kl\pi}{n}|il\rangle\langle jk| \\
&=-\sum_{l=1}^nA_h(t_l)\otimes|l\rangle\langle l|P_n.
\label{eq:matrixL2}
\end{align}
Note that if $A_h$ is time-independent, then
\begin{equation}
L_2(A_h)=-A_h\otimes P_n.
\end{equation}

The matrix $L_3$ combines the Chebyshev series coefficients $c_{i,l}$ to produce $x_i$ for each $i\in\rangez{d}$. To express the final state $x(-1)$, $L_3$ represents the linear combination $x_i(-1)=\sum_{k=0}^nc_{i,k}T_k(-1)=\sum_{k=0}^n(-1)^kc_{i,k}$. Thus we take
\begin{equation}
L_3=\sum_{i=0}^{d-1}\sum_{k=0}^n(-1)^k|i0\rangle\langle ik|.
\end{equation}
Notice that $L_3$ has zero rows for $l\in\range{n}$.

When $h=m$, $L_4$ is used to construct $x_i$ from the output of $L_3$ for $l=0$, and to repeat $x_i$ $n$ times for $l\in\range{n}$. When $m+1\le h\le m+p$, both $L_4$ and $L_5$ are used to repeat $x_i$ $(n+1)p$ times for $l\in\range{n}$. This repetition serves to increase the success probability of the final measurement. In particular, we take
\begin{equation}
L_4=-\sum_{i=0}^{d-1}\sum_{l=1}^n|il\rangle\langle il-1|+\sum_{i=0}^{d-1}\sum_{l=0}^n|il\rangle\langle il|
\end{equation}
and
\begin{equation}
L_5=-\sum_{i=0}^{d-1}|i0\rangle\langle in|.
\end{equation}

In summary, the linear system is as follows. For each $h \in \rangez{m}$, $(L_1+L_2(A_h))|X\rangle=|B_h\rangle$ solves the differential equations over $[\Gamma_h,\Gamma_{h+1}]$, and the coefficients $c_{i,l}(\Gamma_{h+1})$ are combined by $L_3$ into the $(h+1)$st block as initial conditions. When $h=m$, the final coefficients $c_{i,l}(\Gamma_m)$ are combined by $L_3$ and $L_4$ into the final state with coefficients $x_i$, and this solution is repeated $(p+1)(n+1)$ times by $L_4$ and $L_5$.

To explicitly illustrate the structure of this system, we present a simple example in \app{example}.

\section{Solution error}
\label{sec:error}

In this section, we bound how well the solution of the linear system defined above approximates the actual solution of the system of differential equations.

\begin{lemma}\label{lem:error}
For the linear system $L|X\rangle=|B\rangle$ defined in \eq{linear_system}, let $x$ be the approximate ODE solution specified by the linear system and let $\hat{x}$ be the exact ODE solution. For all $t\in[0,T]$, assume $A(t)$ can be diagonalized as $A(t)=V(t)\Lambda(t)V^{-1}(t)$ for some $\Lambda(t)=\diag(\lambda_0(t),\ldots,\lambda_d(t))$ with $\Re(\lambda_i(t))\le0$ for all $i\in\rangez{d}$. Let $\kappa_V:=\max_t\kappa_V(t)$ be an upper bound of the condition number of $V(t)$. Then for $n$ sufficiently large, the error in the solution at time $T$ satisfies
\begin{equation}
\|\hat{x}(T)-x(T)\|\le m\max_{t\in[0,T]}\|\hat{x}^{(n+1)}(t)\|\frac{e^{n+1}}{(2n)^n}.
\end{equation}
\end{lemma}

\begin{proof}
First we carefully choose $n$ satisfying
\begin{equation}
n\ge\frac{e}{2}\left\lfloor\frac{\log(\omega)}{\log(\log(\omega))}\right\rfloor
\end{equation}
where
\begin{equation}
\omega:=\max_{t\in[0,T]}\frac{\|\hat{x}^{(n+1)}(t)\|}{\|\gamma\|}(m+1)
\end{equation}
to ensure that
\begin{equation}
\max_{t\in[0,T]}\frac{\|\hat{x}^{(n+1)}(t)\|}{\|\gamma\|}\Bigl(\frac{e}{2n}\Bigr)^n\le\frac{1}{m+1}.
\label{eq:nm_relation}
\end{equation}

According to the quantum spectral method defined in \sec{linear_system}, we solve
\begin{equation}
\frac{\d{x}}{\d{t}}=A_h(t)x(t)+f_h(t),\quad h\in\rangez{m}.
\label{eq:spilt_linear_system}
\end{equation}
We denote the exact solution by $\hat{x}(\Gamma_{h+1})$, and we let $x(\Gamma_{h+1})=\sum_{i=0}^d\sum_{l=0}^n(-1)^nc_{i,l}(\Gamma_{h+1})$, where $c_{i,l}(\Gamma_{h+1})$ is defined in \eq{vectorX}. Define
\begin{equation}
\Delta_{h+1}:=\|\hat{x}(\Gamma_{h+1})-x(\Gamma_{h+1})\|.
\end{equation}

For $h=0$, \lem{smooth_spectral_approx} implies
\begin{equation}
\Delta_1=\|\hat{x}(\Gamma_1)-x(\Gamma_1)\|\le\max_{t\in[0,T]}\|\hat{x}^{(n+1)}(t)\|\Bigl(\frac{e}{2n}\Bigr)^n.
\end{equation}

For $h \in \range{m}$, the error in the approximate solution of $\frac{\d{x}}{\d{t}}=A_h(t)x(t)+f_h(t)$ has two contributions: the error from the linear system and the error in the initial condition. We let $\widetilde{x}(\Gamma_{h+1})$ denote the solution of the linear system $\bigl(L_1 + L_2(A_h)\bigr) |\widetilde{x}(\Gamma_{h+1})\rangle = |B(f_h)\rangle$ under the initial condition $\hat{x}(\Gamma_h)$.
Then
\begin{equation}
\Delta_{h+1}
\le\|\hat{x}(\Gamma_{h+1})-\widetilde{x}(\Gamma_{h+1})\|+\|\widetilde{x}(\Gamma_{h+1})-x(\Gamma_{h+1})\|.
\end{equation}

The first term can be bounded using \lem{smooth_spectral_approx}, giving
\begin{equation}
\|\hat{x}(\Gamma_{h+1})-\widetilde{x}(\Gamma_{h+1})\|\le\max_{t\in[0,T]}\|\hat{x}^{(n+1)}(t)\|\Bigl(\frac{e}{2n}\Bigr)^n.
\end{equation}
The second term comes from the initial error $\Delta_h$, which is transported through the linear system. Let
\begin{equation}
E_{h+1}=\hat{E}_{h+1}+\delta_{h+1}
\end{equation}
where $E_{h+1}$ is the solution of the linear system with input $\Delta_h$ and $\hat{E}_{h+1}$ is the exact solution of $\frac{\d{x}}{\d{t}}=A_{h+1}(t)x(t)+f_{h+1}(t)$ with initial condition $x(\Gamma_h)=\Delta_h$. Then by \lem{smooth_spectral_approx},
\begin{equation}
\|\delta_{h+1}\|=\|\hat{E}_{h+1}-E_{h+1}\|\le\frac{\Delta_h}{\|\gamma\|}\max_{t\in[0,T]}\|\hat{x}^{(n+1)}(t)\|\Bigl(\frac{e}{2n}\Bigr)^n,
\end{equation}
so
\begin{equation}
\|\widetilde{x}(\Gamma_{h+1})-x(\Gamma_{h+1})\|\le\Delta_h+\frac{\Delta_h}{\|\gamma\|}\max_{t\in[0,T]}\|\hat{x}^{(n+1)}(t)\|\Bigl(\frac{e}{2n}\Bigr)^n.
\end{equation}
Thus, we have an inequality recurrence for bounding $\Delta_h$:
\begin{align}
\Delta_{h+1}&\le
\Bigl(1+\max_{t\in[0,T]}\frac{\|\hat{x}^{(n+1)}(t)\|}{\|\gamma\|}\Bigl(\frac{e}{2n}\Bigr)^n\Bigr)\Delta_h+\max_{t\in[0,T]}\|\hat{x}^{(n+1)}(t)\|\Bigl(\frac{e}{2n}\Bigr)^n.
\end{align}

Now we iterate $h$ from $1$ to $m$. Equation \eq{nm_relation} implies
\begin{equation}
\max_{t\in[0,T]}\frac{\|\hat{x}^{(n+1)}(t)\|}{\|\gamma\|}\Bigl(\frac{e}{2n}\Bigr)^n\le\frac{1}{m+1}\le\frac{1}{m},
\end{equation}
so
\begin{equation}
\Bigl(1+\max_{t\in[0,T]}\frac{\|\hat{x}^{(n+1)}(t)\|}{\|\gamma\|}\Bigl(\frac{e}{2n}\Bigr)^n\Bigr)^{m-1}\le\Bigl(1+\frac{1}{m}\Bigr)^m\le e.
\end{equation}
Therefore
\begin{equation}
\begin{aligned}
\Delta_m
&\le\Bigl(1+\max_{t\in[0,T]}\frac{\|\hat{x}^{(n+1)}(t)\|}{\|\gamma\|}\Bigl(\frac{e}{2n}\Bigr)^n\Bigr)^{m-1}\Delta_1 \\
&\quad+\sum_{h=1}^{m-1}\Bigl(1+\max_{t\in[0,T]}\frac{\|\hat{x}^{(n+1)}(t)\|}{\|\gamma\|}\Bigl(\frac{e}{2n}\Bigr)^n\Bigr)^{h-1}C\max_{t\in[0,T]}\|\hat{x}^{(n+1)}(t)\|\Bigl(\frac{e}{2n}\Bigr)^n\\
&\le\Bigl(1+\frac{1}{m}\Bigr)^{m-1}\Delta_1+(m-1)\Bigl(1+\frac{1}{m}\Bigr)^{m-1}C\max_{t\in[0,T]}\|\hat{x}^{(n+1)}(t)\|\Bigl(\frac{e}{2n}\Bigr)^n\\
&\le\max_{t\in[0,T]}\|\hat{x}^{(n+1)}(t)\|\frac{e^{n+1}}{(2n)^n}+(m-1)\max_{t\in[0,T]}\|\hat{x}^{(n+1)}(t)\|\frac{e^{n+1}}{(2n)^n}\\
&=m\max_{t\in[0,T]}\|\hat{x}^{(n+1)}(t)\|\frac{e^{n+1}}{(2n)^n},
\end{aligned}
\end{equation}
which shows that the solution error decreases exponentially with $n$. In other words, the quantum spectral method approximates the solution with error $\epsilon$ using $n=\poly(\log({1}/{\epsilon}))$.
\end{proof}

Note that for time-independent differential equations, we can directly estimate $\|\hat{x}^{(n+1)}(t)\|$ using
\begin{equation}
\hat{x}^{(n+1)}(t)=A_h^{n+1}\hat{x}(t)+A_h^nf_h.
\label{eq:xhat_derivatives}
\end{equation}
Writing $A_h=V_h\Lambda_hV_h^{-1}$ where $\Lambda_h=\diag(\lambda_0,\ldots,\lambda_{d-1})$, we have $e^{A_h}=V_he^{\Lambda_h}V_h^{-1}$. Thus the exact solution of time-independent equation with initial condition $\hat{x}(1)=\gamma$ is
\begin{equation}
\begin{aligned}
\hat{x}(t)
&=e^{A_h(1-t)}\gamma+(e^{A_h(1-t)}-I)A_h^{-1}f_h\\
&=V_he^{\Lambda_h}V_h^{-1}\gamma+V_h(e^{\Lambda_h(1-t)}-I)\Lambda_h^{-1}V_h^{-1}f_h.
\end{aligned}
\end{equation}
Since $\Re(\lambda_i) \le 0$ for all eigenvalues $\lambda_i$ of $A_h$ for $i \in \rangez{d}$, we have $\|e^{\Lambda_h}\|\le1$. Therefore
\begin{equation}
\|\hat{x}(t)\|\le\kappa_V(\|\gamma\|+2\|f_h\|).
\end{equation}
Furthermore, since  $\max_{h,t}\|A_h(t)\|\le1$, we have
\begin{equation}
\begin{aligned}
\max_{t\in[0,T]}\|\hat{x}^{(n+1)}(t)\|&\le\max_{t\in[0,T]}(\|\hat{x}(t)\|+\|f_h(t)\|)\\
&\le\kappa_V(\|\gamma\|+3\|f_h\|)\\
&\le\kappa_V(\|\gamma\|+2\tau\|f\|).\label{eq:time_indep_bound}
\end{aligned}
\end{equation}
Thus the solution error satisfies
\begin{equation}
\|\hat{x}(T)-x(T)\|\le m\kappa_V(\|\gamma\|+2\tau\|f\|)\frac{e^{n+1}}{(2n)^n}.
\end{equation}
Note that, although we represent the solution differently, this bound is similar to the corresponding bound in \cite[Theorem 6]{BCOW17}.

\section{Condition number}
\label{sec:condition_number}

We now analyze the condition number of the linear system.

\begin{lemma}\label{lem:condition_number}
Consider an instance of the quantum ODE problem as defined in \prb{ode}. For all $t\in[0,T]$, assume $A(t)$ can be diagonalized as $A(t)=V(t)\Lambda(t)V^{-1}(t)$ for some $\Lambda(t)=\diag(\lambda_0(t),\ldots,\lambda_d(t))$, with $\Re(\lambda_i(t))\le0$ for all $i\in\rangez{d}$. Let $\kappa_V:=\max_{t\in[0,T]}\kappa_V(t)$ be an upper bound on the condition number of $V(t)$. Then for $m,p\in\Z^+$ and $n$ sufficiently large, the condition number of $L$ in the linear system \eq{linear_system} satisfies
\begin{equation}
\kappa_L\le(\pi m+p+2)(n+1)^{3.5}(2\kappa_V+e\|\gamma\|).
\end{equation}
\end{lemma}

\begin{proof}
We begin by bounding the norms of some operators that appear in the definition of $L$.
First we consider the  $l_{\infty}$ norm of $D_n$ since this is straightforward to calculate:
\begin{equation}
\Vert D_n\Vert_{\infty}:=\max_{1\le i\le n}\sum_{j=0}^n|[D_n]_{ij}|=
\begin{cases}
\frac{n(n+2)}{2}    &\text{$n$ even},\\
\frac{(n+1)^2}{2}-2 &\text{$n$ odd}.
\end{cases}
\end{equation}
Thus we have the upper bound
\begin{equation}
\Vert D_n\Vert\le\sqrt{n+1}\Vert D_n\Vert_{\infty}\le\frac{(n+1)^{2.5}}{2}.
\end{equation}

Next we upper bound the spectral norm of the discrete cosine transform matrix $P_n$:
\begin{equation}
\begin{aligned}
\Vert P_n\Vert^2
&\le\max_{0\le l\le n}\sum_{k=0}^n\cos^2\frac{kl\pi}{n}
=\max_{0\le l\le n}\sum_{k=0}^n\biggl(\frac{1}{2}+\cos\frac{2kl\pi}{n}\biggr)\\
&=\frac{n+1}{2}+\max_{0\le l\le n}\sum_{k=0}^n\cos\frac{2kl\pi}{n}
\le\frac{n+1}{2}+1.
\end{aligned}
\end{equation}
Therefore
\begin{equation}
\Vert P_n\Vert\le\sqrt{n+1}.
\end{equation}
Thus we can upper bound the norm of $L_1$ as
\begin{equation}
\Vert L_1\Vert\le\Vert D_n\Vert\Vert P_n\Vert\le\frac{(n+1)^3}{2}.
\end{equation}

Next we consider the spectral norm of $L_2(A_h)$ for any $h \in \rangez{m}$. We have
\begin{equation}
L_2(A_h)=-\sum_{l=1}^nA_h(t_l)\otimes|l\rangle\langle l|P_n.
\end{equation}
Since the eigenvalues of each $A_h(t_l)$ for $l \in \rangez{n+1}$ are all eigenvalues of
\begin{equation}
\sum_{l=0}^nA_h(t_l)\otimes|l\rangle\langle l|,
\end{equation}
we have
\begin{equation}
\biggl\|\sum_{l=1}^nA_h(t_l)\otimes|l\rangle\langle l|\biggr\|
\le\biggl\|\sum_{l=0}^nA_h(t_l)\otimes|l\rangle\langle l|\biggr\|
\le\max_{t \in [-1,1]}\|A_h(t)\|\le1
\end{equation}
by \eq{tau}. Therefore
\begin{equation}
\Vert L_2(A_h)\Vert\le\Vert P_n\Vert\le\sqrt{n+1}.
\end{equation}

By direct calculation, we have
\begin{align}
\Vert L_3\Vert&\le n+1, \\
\Vert L_4\Vert&=2, \\
\Vert L_5\Vert&=1.
\end{align}
Thus, for $n\ge5$, we find
\begin{equation}
\Vert L\Vert\le\frac{(n+1)^3}{2}+\sqrt{n+1}+n+4\le(n+1)^3.\
\label{eq:Lnorm}
\end{equation}

Next we upper bound $\Vert L^{-1}\Vert$. By definition,
\begin{equation}
\Vert L^{-1}\Vert=\sup_{\||B\rangle\|\le1}\|L^{-1}|B\rangle\|.
\end{equation}
We express $|B\rangle$ as
\begin{equation}
|B\rangle
=\sum_{h=0}^{m+p}\sum_{l=0}^n \sum_{i=0}^{d-1}\beta_{hil}|hil\rangle
=\sum_{h=0}^{m+p}\sum_{l=0}^n |b_{hl}\rangle
\label{eq:Bvec}
\end{equation}
where $|b_{hl}\rangle := \sum_{i=0}^{d-1}\beta_{hil}|hil\rangle$ satisfies $\||b_{hl}\rangle\|^2=\sum_{i=0}^{d-1}|\beta_{hil}|^2\le1$. For any fixed $h \in \rangez{m+p+1}$ and $l \in \rangez{n+1}$, we first upper bound $\|L^{-1}|b_{hl}\rangle\|$ and use this to upper bound the norm of $L^{-1}$ applied to linear combinations of such vectors.

Recall that the linear system comes from \eq{linear_eq}, which is equivalent to
\begin{equation}
\sum_{k=0}^nT'_k(t_r)c_{i,k}(T_h)=\sum_{j=0}^{d-1}A_h(t_r)_{ij}\sum_{k=0}^nT_k(t_r)c_{j,k}(T_h)+f_h(t_r)_i,\quad i\in\rangez{d},~r\in\rangez{n+1}.
\label{eq:equiv_system}
\end{equation}
For fixed $h \in \rangez{m+p+1}$ and $r \in \rangez{n+1}$, define vectors $x_{hr},x_{hr}' \in \C^d$ with
\begin{equation}
(x_{hr})_i := \sum_{k=0}^nT_k(t_r)c_{i,k}(T_h), \qquad
(x'_{hr})_i:= \sum_{k=0}^nT'_k(t_r)c_{i,k}(T_h)
\end{equation}
for $i \in \rangez{d}$.
We claim that $x_{hr} = x_{hr}' = 0$ for any $r \ne l$.
Combining only the equations from \eq{equiv_system} with $r \ne l$ gives the system
\begin{equation}
x'_{hr} = A_h(t_r)x_{hr} .
\label{eq:ori_diffeq}
\end{equation}
Consider a corresponding system of differential equations
\begin{align}
  \frac{\d{\hat x_{hr}(t)}}{\d{t}} = A_h(t_r) \hat x(t) + b
  \label{eq:corr_diffeq}
\end{align}
where $\hat x_{hr}(t) \in \C^d$ for all $t \in [-1,1]$. The solution of this system with $b=0$ and initial condition $\hat x_{hr}(1) = 0$ is clearly $\hat x_{hr}(t)=0$ for all $t \in [-1,1]$. Then the $n$th-order truncated Chebyshev approximation of \eq{corr_diffeq}, which should satisfy the linear system \eq{ori_diffeq} by \eq{cheby_expand} and \eq{cheby_ode}, is exactly $x_{hr}$. Using \lem{error} and observing that $\hat{x}^{(n+1)}(t)=0$, we have
\begin{equation}
x_{hr}=\hat x_{hr}(t)=0.
\end{equation}

When $t=t_l$, we let $|B\rangle=|b_{hl}\rangle$ denote the first nonzero vector. Combining only the equations from \eq{equiv_system} with $r = l$ gives the system
\begin{equation}
x'_{hl} = A_h(t_l)x_{hl} .
\label{eq:ori_diffeq_l}
\end{equation}
Consider a corresponding system of differential equations
\begin{align}
  \frac{\d{\hat x_{hr}(t)}}{\d{t}} = A_h(t_r) \hat x(t) + b,
  \label{eq:corr_diffeq_l}
\end{align}
with $\gamma=b_{h0},b=0$ for $l=0$; or $\gamma=0,b=b_{hl}$ for $l\in\range{n}$.

Using the diagonalization $A_h(t_l)=V(t_l)\Lambda_h(t_l)V^{-1}(t_l)$, we have $e^A=V(t_l)e^{\Lambda_h(t_l)}V^{-1}(t_l)$. Thus the exact solution of the differential equations \eq{corr_diffeq} with $r=l$ and initial condition $\hat x_{hr}(1)=\gamma$ is
\begin{equation}
\begin{aligned}
\hat{x}_{hr}(t)
&=e^{A_h(t_l)(1-t)}\gamma+(e^{A_h(t_l)(1-t)}-I)A_h(t_l)^{-1}b\\
&=V(t_l)e^{\Lambda_h(t_l)(1-t)}V^{-1}(t_l)\gamma+V(e^{\Lambda_h(t_l)(1-t)}-I)\Lambda_h(t_l)^{-1}V^{-1}b.
\end{aligned}
\label{eq:exact_diffeq_soln}
\end{equation}

According to equation \eq{nm_relation} in the proof of \lem{error}, we have
\begin{equation}
x_{hl}=\hat{x}_{hl}(-1)+\delta_{hl}
\label{eq:approx_error}
\end{equation}
where
\begin{equation}
\|\delta_{hl}\|\le\max_{t\in[0,T]}\|\hat{x}_{hl}^{(n+1)}(t)\|\frac{e^{n+1}}{(2n)^n}\le\frac{e\|\gamma\|}{m+1}.
\end{equation}

Now for $h \in \rangez{m+1}$, we take $x_{hl}$ to be the initial condition $\gamma$ for the next subinterval to obtain $x_{(h+1)l}$. Using \eq{exact_diffeq_soln} and \eq{approx_error}, starting from $\gamma=b_{h0},b=0$ for $l=0$, we find
\begin{equation}
x_{ml}
=V(t_l) \biggl(\prod_{j=1}^{m-h+1}e^{2\Lambda_h(t_l)}\biggr)V^{-1}(t_l)\gamma
+\sum_{k=0}^{m-h}V(t_l)\biggl(\prod_{j=1}^k e^{2\Lambda_h(t_l)}\biggr)V^{-1}(t_l)\delta_{(m-k)l}.
\end{equation}
Since $\|\Lambda_h(t_l)\|\le\|\Lambda\|\le1$ and $\Lambda_h(t_l)=\diag(\lambda_0,\ldots,\lambda_{d-1})$ with $\Re(\lambda_i)\le0$ for $i \in \rangez{d}$, we have $\|e^{2\Lambda_h(t_l)}\|\le1$. Therefore
\begin{align}
\|x_{hl}\|
&\le\|x_{ml}\|
\le \kappa_V(t_l)\|b_{hl}\|+(m-h+1)\kappa_V(t_l)\|\delta_{hl}\|
\le\kappa_V(t_l)+e\|\gamma\|
\le\kappa_V+e\|\gamma\|.
\end{align}
On the other hand, with $\gamma=0, b=b_{hl}$ for $l\in\range{n}$, we have
\begin{equation}
  \begin{aligned}
x_{ml}
&=V(t_l)\biggl(\prod_{j=1}^{m-h}e^{2\Lambda_h(t_l)}\biggr)\bigl((e^{2\Lambda_h(t_l)}-I)\Lambda_h(t_l)^{-1}\bigr)V^{-1}(t_l)b \\
&\quad+\sum_{k=0}^{m-h}V(t_l)\biggl(\prod_{j=1}^ke^{2\Lambda_h(t_l)}\biggr)V^{-1}(t_l)\delta_{(m-k)l},
  \end{aligned}
\end{equation}
so
\begin{equation}
\|x_{hl}\|
\le2\kappa_V(t_l)\|b_{hl}\|+(m-h+1)\kappa_V(t_l)\|\delta_{hl}\|
\le2\kappa_V(t_l)+e\|\gamma\|
\le2\kappa_V+e\|\gamma\|.
\end{equation}
For $h \in \{m,m+1,\ldots,m+p\}$, according to the definition of $L_4$ and $L_5$, we similarly have
\begin{equation}
\|x_{hl}\|=\|x_{ml}\|\le2\kappa_V+e\|\gamma\|. 
\end{equation}

According to \eq{exact_diffeq_soln}, $\hat{x}_{hl}(t)$ is a monotonic function of $t \in [-1,1]$, which implies
\begin{equation}
\|\hat{x}_{hl}(t)\|^2\le\max\{\|\hat{x}_{hl}(-1)\|^2,\|\hat{x}_{hl}(1)\|^2\}\le(2\kappa_V+e\|\gamma\|)^2.
\end{equation}
Using the identity
\begin{equation}
\int_{-1}^1\frac{\d{t}}{\sqrt{1-t^2}}=\pi,
\label{eq:weighted_integration}
\end{equation}
we have
\begin{equation}
\int_{-1}^1\|\hat{x}_{hl}(t)\|^2\frac{\d{t}}{\sqrt{1-t^2}} \le (2\kappa_V+e\|\gamma\|)^2\int_{-1}^1\frac{\d{t}}{\sqrt{1-t^2}}=\pi(2\kappa_V+e\|\gamma\|)^2.
\label{eq:weighted_inequality}
\end{equation}
Consider the Chebyshev expansion of $\hat{x}_{hl}(t)$ as in \eq{cheby_expand}:
\begin{equation}
\hat{x}_{hl}(t)=\sum_{i=0}^{d-1}\sum_{l=0}^{\infty}c_{i,l}(T_{h+1})T_l(\overline{t}).
\end{equation}
By the orthogonality of Chebyshev polynomials (as specified in \eq{cheby_orth}), we have
\begin{equation}
\begin{aligned}
&\int_{-1}^1\|\hat{x}_{hl}(t)\|^2\frac{\d{t}}{\sqrt{1-t^2}}=\int_{-1}^1 \biggl(\sum_{i=0}^{d-1}\sum_{l=0}^{\infty}c_{i,l}(T_{h+1})T_l(\overline{t})\biggr)^2\frac{\d{t}}{\sqrt{1-t^2}}\\
=&\sum_{i=0}^{d-1}\sum_{l=1}^{\infty}c^2_{i,l}(T_{h+1})+2\sum_{i=0}^{d-1}c^2_{i,0}(T_{h+1})\ge\sum_{i=0}^{d-1}\sum_{l=1}^nc^2_{i,l}(T_{h+1})+2\sum_{i=0}^{d-1}c^2_{i,0}(T_{h+1}).
\end{aligned}
\end{equation}
Using \eq{weighted_inequality}, this gives
\begin{equation}
\sum_{i=0}^{d-1}\sum_{l=0}^nc^2_{i,l}(T_{h+1}) \le \int_{-1}^1\hat{x}_{hl}^2(t)\frac{\d{t}}{\sqrt{1-t^2}} \le \pi(2\kappa_V+e\|\gamma\|)^2.
\end{equation}

Now we compute $\Vert|X\rangle\Vert$, summing the contributions from all $c_{i,r}(\Gamma_h)$ and $x_{mr}$, and notice that $c_{i,r}=0$ and $x_{mr}=0$ for all $r\ne l$, giving
\begin{equation}
\begin{aligned}
\||X\rangle\|^2&=\sum_{h=0}^{m-1}\sum_{i=0}^{d-1}c^2_{i,l}(T_{h+1})+(p+1)(x_{ml})^2 \\
&\le \pi m(2\kappa_V+e\|\gamma\|)^2+(p+1)(\kappa_V+e\|\gamma\|)^2 \\
&\le (\pi m+p+1)(2\kappa_V+e\|\gamma\|)^2.
\end{aligned}
\end{equation}

Finally, considering all $h\in\rangez{m+p+1}$ and $l\in\rangez{n+1}$, from \eq{Bvec} we have
\begin{equation}
\||B\rangle\|^2=\sum_{h=0}^{m+p}\sum_{l=0}^n\||b_{hl}\rangle\|^2\le1,
\end{equation}
so
\begin{equation}
\begin{aligned}
\| L^{-1}\|^2
&=\sup_{\||B\rangle\|\le1}\| L^{-1}|B\rangle\|^2
=\sup_{\||B\rangle\|\le1}\sum_{h=0}^{m+p}\sum_{l=0}^n\|L^{-1}|b_{hl}\rangle\|^2\\
&\le(\pi m+p+1)(m+p+1)(n+1)(2\kappa_V+e\|\gamma\|)^2\\
&\le(\pi m+p+1)^2(n+1)(2\kappa_V+e\|\gamma\|)^2,
\end{aligned}
\end{equation}
and therefore
\begin{equation}
\| L^{-1}\|\le(\pi m+p+1)(n+1)^{0.5}(2\kappa_V+e\|\gamma\|).
\label{eq:Linvnorm}
\end{equation}

Finally, combining \eq{Lnorm} and \eq{Linvnorm} gives
\begin{equation}
\kappa_L=\| L\|\| L^{-1}\|\le(\pi m+p+1)(n+1)^{3.5}(2\kappa_V+e\|\gamma\|)
\end{equation}
as claimed.
\end{proof}

\section{Success probability}
\label{sec:success_prob}

We now evaluate the success probability of our approach to the quantum ODE problem.

\begin{lemma}\label{lem:success_prob}
Consider an instance of the quantum ODE problem as defined in \prb{ode} with the exact solution $\hat{x}(t)$ for $t\in [0,T]$, and its corresponding linear system \eq{linear_system} with $m,p\in\Z^+$ and $n$ sufficiently large. When applying the QLSA to this system, the probability of measuring a state proportional to $|x(T)\rangle=\sum_{i=0}^{d-1}x_i|i\rangle$ is
\begin{equation}
P_{\mathrm{measure}} \ge \frac{(p+1)(n+1)}{\pi mq^2+(p+1)(n+1)},
\end{equation}
where $x_i$ is defined in \eq{vectorX}, $\tau$ is defined in \eq{tau_max}, and
\begin{equation}
q:=\max_{t\in [0,T]} \frac{\|\hat{x}(t)\|}{\|x(T)\|}.
\end{equation}
\end{lemma}

\begin{proof}
After solving the linear system \eq{linear_system} using the QLSA, we measure the first and third registers of $|X\rangle$ (as defined in \eq{vectorX}). We decompose this state as
\begin{equation}
|X\rangle=|X_{\mathrm{bad}}\rangle+|X_{\mathrm{good}}\rangle,
\end{equation}
where
\begin{align}
|X_{\mathrm{bad}}\rangle&=\sum_{h=0}^{m-1}\sum_{i=0}^{d-1}\sum_{l=0}^nc_{i,l}(\Gamma_{h+1})|hil\rangle, \\
|X_{\mathrm{good}}\rangle&=\sum_{h=m}^{m+p}\sum_{i=0}^{d-1}\sum_{l=0}^nx_i|hil\rangle.
\end{align}

When the first register is observed in some $h \in \{m,m+1,\ldots,m+p\}$ (no matter what outcome is seen for the third register), we output the second register, which is then in a normalized state proportional to the final state:
\begin{equation}
|X_{\mathrm{measure}}\rangle=\frac{|x(T)\rangle}{\||x(T)\rangle\|},
\end{equation}
with
\begin{equation}
|x(T)\rangle=\sum_{i=0}^{d-1}x_i|i\rangle=\sum_{i=0}^{d-1}\sum_{k=0}^nc_{i,k}T_k(t)|i\rangle.
\end{equation}

Notice that
\begin{equation}
\||x(T)\rangle\|^2=\sum_{i=0}^{d-1}x_i^2
\end{equation}
and
\begin{equation}
\||X_{\mathrm{good}}\rangle\|^2=(p+1)(n+1)\sum_{i=0}^{d-1}x_i^2=(p+1)(n+1)\||x(T)\rangle\|^2.
\end{equation}

Considering the definition of $q$, the contribution from time interval $h$ under the rescaling \eq{rescaling_map}, and the identity \eq{weighted_integration}, we have
\begin{equation}
\begin{aligned}
q^2\|x(T)\|^2
&=\max_{t\in [0,T]}\|\hat{x}(t)\|^2
=\frac{1}{\pi}\int_{-1}^1 \frac{\d{\tau}}{\sqrt{1-\tau^2}}\max_{t\in [0,T]}\|\hat{x}(t)\|^2
\ge\frac{1}{\pi}\int_{-1}^1 \frac{\d{\tau}}{\sqrt{1-\tau^2}}\max_{t\in [\Gamma_h,\Gamma_{h+1}]}\|\hat{x}(t)\|^2\\
&=\frac{1}{\pi}\int_{-1}^1 \frac{\d{\tau}}{\sqrt{1-\tau^2}}\max_{\overline{t}\in [-1,1]}\|\hat{x}_h(\overline{t})\|^2
\ge\frac{1}{\pi}\int_{-1}^1 \|\hat{x}_h(\overline{t})\|^2 \frac{\d{\overline{t}}}{\sqrt{1-\overline{t}^2}},
\end{aligned}
\end{equation}
where $\hat{x}_h(\overline{t})$ is the solution of \eq{spilt_linear_system} with the rescaling in \eq{rescaling_rate}. By the orthogonality of Chebyshev polynomials (as specified in \eq{cheby_orth}),
\begin{equation}
\begin{aligned}
q^2\|x(T)\|^2
&\ge\frac{1}{\pi}\int_{-1}^1 \|\hat{x}_h(\overline{t})\|^2 \frac{\d{\overline{t}}}{\sqrt{1-\overline{t}^2}}
=\frac{1}{\pi}\int_{-1}^1(\sum_{i=0}^{d-1}\sum_{k=0}^{\infty}c_{i,k}(T_{h+1})T_k(\overline{t}))^2 \frac{\d{\overline{t}}}{\sqrt{1-\overline{t}^2}}\\
&=\frac{1}{\pi}(\sum_{i=0}^{d-1}\sum_{k=1}^{\infty}c^2_{i,k}(T_{h+1})+2\sum_{i=0}^{d-1}c^2_{i,0}(T_{h+1}))
\ge\frac{1}{\pi}\sum_{i=0}^{d-1}\sum_{k=0}^nc^2_{i,k}(T_{h+1}).
\end{aligned}
\end{equation}
For all $h\in\rangez{m}$, we have
\begin{equation}
mq^2\|x(T)\|^2\ge\sum_{h=0}^{m-1}\frac{1}{\pi}\sum_{i=0}^{d-1}\sum_{k=0}^nc^2_{i,k}(T_{h+1})
=\frac{1}{\pi}\||X_{\mathrm{bad}}\rangle\|^2,
\end{equation}
and therefore
\begin{equation}
\||X_{\mathrm{good}}\rangle\|^2
=(p+1)(n+1)\|x(T)\|^2
\ge \frac{(p+1)(n+1)}{\pi mq^2}\||X_{\mathrm{bad}}\rangle\|^2.
\end{equation}

Thus we see that the success probability of the measurement satisfies
\begin{equation}
P_{\mathrm{measure}} \ge \frac{(p+1)(n+1)}{\pi mq^2+(p+1)(n+1)}
\end{equation}
as claimed.
\end{proof}

\section{State preparation}
\label{sec:state preparation}

We now describe a procedure for preparing the vector $|B\rangle$ in the linear system \eq{linear_system} (defined in \eq{vectorB} and \eq{B_f_h}) using the given ability to prepare the initial state of the system of differential equations. We also evaluate the complexity of this procedure.

\begin{lemma}\label{lem:preparation}
Consider state preparation oracles acting on a state space with basis vectors $|h\rangle|i\rangle|l\rangle$ for $h\in\rangez{m},i\in\rangez{d},l\in\rangez{n}$, where $m,d,n \in \N$, encoding an initial condition $\gamma \in \C^d$ and function $f_{h}(\cos\frac{l\pi}{n}) \in \C^d$ as in \eq{B_f_h}. Specifically, for any $h\in\rangez{m}$ and $l\in\range{n}$, let $O_x$ be a unitary oracle that maps $|0\rangle|0\rangle|0\rangle$ to a state proportional to $|0\rangle|\gamma\rangle|0\rangle$ and $|h\rangle|\phi\rangle|l\rangle$ to $|h\rangle|\phi\rangle|l\rangle$ for any $|\phi\rangle$ orthogonal to $|0\rangle$; let $O_f(h,l)$ be a unitary that maps $|h\rangle|0\rangle|l\rangle$ to a state proportional to $|h\rangle|f_{h}(\cos\frac{l\pi}{n})\rangle|l\rangle$ and maps $|0\rangle|\phi\rangle|0\rangle$ to $|0\rangle|\phi\rangle|0\rangle$ for and $|\phi\rangle$ orthogonal to $|0\rangle$. Suppose $\|\gamma\|$ and $\|f_{h}(\cos\frac{l\pi}{n})\|$ are known. Then the normalized quantum state
\begin{equation}
  |B\rangle \propto
|0\rangle|\gamma\rangle|0\rangle+\sum_{h=0}^{m-1}\sum_{l=1}^n|h\rangle|f_{h}(\cos\tfrac{l\pi}{n})\rangle|l\rangle
\end{equation}
can be prepared with gate and query complexity $O(mn)$.
\end{lemma}

\begin{proof}
We normalize the components of the state using the coefficients
\begin{equation}
\begin{aligned}
&b_{00}=\frac{\|\gamma\|}{\sqrt{\|\gamma\|^2+\sum_{l=1}^n\| f_h(\cos\frac{l\pi}{n})\|^2}},\\
&b_{hl}=\frac{\| f_h(\cos\frac{l\pi}{n})\|}{\sqrt{\|\gamma\|^2+\sum_{l=1}^n\| f_h(\cos\frac{l\pi}{n})\|^2}}, \quad h\in\rangez{m}, l\in\range{n}
\end{aligned}
\end{equation}
so that
\begin{equation}
  \sum_{h=0}^{m-1}\sum_{l=0}^nb_{hl}^2=1.
\end{equation}

First we perform a unitary transformation mapping
\begin{equation}
|0\rangle|0\rangle|0\rangle \mapsto
b_{00}|0\rangle|0\rangle|0\rangle+b_{01}|0\rangle|0\rangle|1\rangle+\cdots+b_{(m-1)n}|m-1\rangle|0\rangle|n\rangle.
\label{eq:stateprep}
\end{equation}
This can be done in time complexity $O(mn)$ by standard techniques \cite{SBM06}.
Then we perform $O_x$ and $O_f(h,l)$ for all $h\in\rangez{m}, l\in\range{n}$, giving
\begin{equation}
|0\rangle|\gamma\rangle|0\rangle+\sum_{h=0}^{m-1}\sum_{l=1}^n|h\rangle|f_{h}(\cos\tfrac{l\pi}{n})\rangle|l\rangle
\end{equation}
using $O(mn)$ queries.
\end{proof}

\section{Main result}
\label{sec:main}

Having analyzed the solution error, condition number, success probability, and state preparation procedure for our approach, we are now ready to establish the main result.

\begin{theorem}\label{thm:main}
Consider an instance of the quantum ODE problem as defined in \prb{ode}.
Assume $A(t)$ can be diagonalized as $A(t)=V(t)\Lambda(t)V^{-1}(t)$ where $\Lambda(t)=\diag(\lambda_1(t),\ldots,\lambda_d(t))$ with $\Re(\lambda_i(t))\le0$ for each $i \in \rangez{d}$ and $t\in[0,T]$.
Then there exists a quantum algorithm that produces a state $x(T)/\| x(T)\|$ $\epsilon$-close to $\hat{x}(T)/\| \hat{x}(T)\|$ in $l^2$ norm, succeeding with probability $\Omega(1)$, with a flag indicating success, using
\begin{equation}
O\bigl(\kappa_Vs\|A\|Tq \, \poly(\log(\kappa_Vs\|A\|g'T/\epsilon g))\bigr)
\end{equation}
queries to oracles $O_A(h,l)$ (a sparse matrix oracle for $A_h(t_l)$ as defined in \eq{Ah}) and $O_x$ and $O_f(h,l)$ (as defined in \lem{preparation}).
Here $\|A\| := \max_{t \in [0,T]}\|A(t)\|$; $\kappa_V:=\max_t\kappa_V(t)$, where $\kappa_V(t)$ is the condition number of $V(t)$; and
\begin{align}
g&:=\|\hat{x}(T)\|, &
g'&:=\max_{t\in [0,T]} \max_{n \in \N}\|\hat{x}^{(n+1)}(t)\|, &
q&:=\max_{t\in [0,T]} \frac{\|\hat{x}(t)\|}{\|x(T)\|}.
\label{eq:main_thm_parameters}
\end{align}
The gate complexity is larger than the query complexity by a factor of $\poly(\log(\kappa_Vds\|A\|g'T/\epsilon))$.
\end{theorem}

\begin{proof}
We first present the algorithm and then analyze its complexity.

\smallskip\noindent
\textbf{Statement of the algorithm.} First, we choose $m$ to guarantee
\begin{equation}
\frac{\|A\|T}{2m}\le1.
\end{equation}
Then, as in \sec{linear_system}, we divide the interval $[0,T]$ into small subintervals $[0,\Gamma_1],[\Gamma_1,\Gamma_2],\ldots,[\Gamma_{m-1},T]$ with $\Gamma_0=0, \Gamma_m=T$, and define
\begin{align}
\tau&:=\max_{0\le h\le m-1}\{\tau_h\}, &
\tau_h&:=|\Gamma_{h+1}-\Gamma_h|=\frac{T}{m}.
\end{align}

Each subinterval $[\Gamma_h,\Gamma_{h+1}]$ for $h\in\range{m-1}$ is mapped onto $[-1,1]$ with a linear mapping $K_h$ satisfying $K_h(\Gamma_h)=1,K_h(\Gamma_{h+1})=-1$:
\begin{equation}
K_h\colon t \mapsto \overline{t}=1-\frac{2(t-\Gamma_h)}{\Gamma_{h+1}-\Gamma_h}.
\end{equation}
We choose
\begin{equation}
n=\frac{e}{2}\max\biggl\{\biggl\lfloor\frac{\log(\Omega)}{\log(\log(\Omega))}\biggr\rfloor,\biggl\lfloor\frac{\log(\omega)}{\log(\log(\omega))}\biggr\rfloor\biggr\}
\end{equation}
where
\begin{equation}
\Omega:=\frac{g'em}{\delta}=\frac{g'em(1+\epsilon)}{g\epsilon}
\label{eq:Omega}
\end{equation}
and
\begin{equation}
\omega:=\frac{g'}{\|\gamma\|}(m+1).
\label{eq:omega}
\end{equation}

Since $\max_{t\in[0,T]}\|\hat{x}^{(n+1)}(t)\|\le g'$ and \lem{error}, this choice guarantees
\begin{equation}
\|\hat{x}(T)-x(T)\|\le m\max_{t\in[0,T]}\|\hat{x}^{(n+1)}(t)\|\frac{e^{n+1}}{(2n)^n}\le\delta
\label{eq:n_inequality_1}
\end{equation}
and
\begin{equation}
\max_{t\in[0,T]}\frac{\|\hat{x}^{(n+1)}(t)\|}{\|\gamma\|}\Bigl(\frac{e}{2n}\Bigr)^n\le\frac{1}{m+1}.
\label{eq:n_inequality_2}
\end{equation}

Now $\|\hat{x}(T)-x(T)\|\le\delta$ implies
\begin{equation}
\biggl\|\frac{\hat{x}(T)}{\|\hat{x}(T)\|}-\frac{x(T)}{\|x(T)\|}\biggr\|
\le\frac{\delta}{\min\{\|\hat{x}(T)\|,\|x(T)\|\}}\le\frac{\delta}{g-\delta}=:\epsilon,
\label{eq:normalized_inequality}
\end{equation}
so we can choose such $n$ to ensure that the normalized output state is $\epsilon$-close to $\hat{x}(T)/\| \hat{x}(T)\|$.

Following \sec{linear_system}, we build the linear system $L|X\rangle=|B\rangle$ (see \eq{linear_system}) that encodes the quantum spectral method.
By \lem{condition_number}, the condition number of this linear system is at most $(\pi m+p+1)(n+1)^{3.5}(2\kappa_V+e\|\gamma\|_{\infty})$. Then we use the QLSA from reference \cite{CKS15} to obtain a normalized state $|X\rangle$ and measure the first and third register of $|X\rangle$ in the standard basis. If the measurement outcome for the first register belongs to
\begin{equation}
S=\{m,m+1,\ldots,m+p\},
\end{equation}
we output the state of the second register, which is a normalized state $| x(T)\rangle/\|| x(T)\rangle\|$ satisfying \eq{normalized_inequality}. By \lem{success_prob}, the probability of this event happening is at least $\frac{(p+1)(n+1)}{\pi mq^2+(p+1)(n+1)}$. To ensure $m+p=O(\|A\|T)$, we can choose
\begin{equation}
p=O(m)=O(\|A\|T),
\end{equation}
so we can achieve success probability $\Omega(1)$ with $O(q/\sqrt{n})$ repetitions of the above procedure.

\smallskip\noindent
\textbf{Analysis of the complexity.} The matrix $L$ is an $(m+p+1)d(n+1)\times(m+p+1)d(n+1)$ matrix with $O(ns)$ nonzero entries in any row or column. By \lem{condition_number} and our choice of parameters, the condition number of $L$ is $O\bigl(\kappa_V(m+p)n^{3.5}\bigr)$. Consequently, by Theorem 5 of \cite{CKS15}, the QLSA produces the state $|x(T)\rangle$ with
\begin{equation}
O\bigl(\kappa_V(m+p)n^{4.5}s\, \poly(\log(\kappa_Vmns/\delta))\bigr)=O\bigl(\kappa_Vs\|A\|T\, \poly(\log(\kappa_Vs\|A\|g'T/\epsilon g))\bigr)
\end{equation}
queries to the oracles $O_A(h,l)$, $O_x$, and $O_f(h,l)$, and its gate complexity is larger by a factor of $\poly(\log(\kappa_Vmnds/\delta))$. Using $O(q/\sqrt{n})$ steps of amplitude amplification to achieve success probability $\Omega(1)$, the overall query complexity of our algorithm is
\begin{equation}
O\bigl(\kappa_V(m+p)n^4sq \, \poly(\log(\kappa_Vmns/\delta))\bigr)=O\bigl(\kappa_Vs\|A\|Tq \, \poly(\log(\kappa_Vs\|A\|g'T/\epsilon g))\bigr),
\end{equation}
and the gate complexity is larger by a factor of
\begin{equation}
\poly(\log(\kappa_Vds\|A\|g'T/\epsilon g))
\end{equation}
as claimed.
\end{proof}

In general, $g'$ could be unbounded above as $n\rightarrow\infty$. However, we could obtain a useful bound in such a case by solving the implicit equations \eq{n_inequality_1} and \eq{n_inequality_2}.

Note that for time-independent differential equations, we can replace $g'$ by $\|\gamma\|+2\tau\|f\|$ as shown in \eq{time_indep_bound}. In place of \eq{Omega} and \eq{omega}, we choose
\begin{equation}
\Omega:=\frac{(\|\gamma\|+2\tau\|f\|)em\kappa_V}{\delta}=\frac{(\|\gamma\|+2\tau\|f\|)em\kappa_V(1+\epsilon)}{g\epsilon}
\end{equation}
and
\begin{equation}
\omega:=\frac{\|\gamma\|+2\tau\|f\|}{\|\gamma\|}(m+1)\kappa_V.
\end{equation}
By \lem{error}, this choice guarantees
\begin{equation}
\|\hat{x}(T)-x(T)\|\le\max_{t\in[-1,1]}\|\hat{x}(t)-x(t)\|\le m\kappa_V(\|\gamma\|+2\tau\|f\|)\frac{e^{n+1}}{(2n)^n}\le\delta
\end{equation}
and
\begin{equation}
\max_{t\in[0,T]}\frac{\|\hat{x}^{(n+1)}(t)\|}{\|\gamma\|}\Bigl(\frac{e}{2n}\Bigr)^n\le\frac{\kappa_V(\|\gamma\|+2\tau\|f\|)}{\|\gamma\|}\Bigl(\frac{e}{2n}\Bigr)^n\le\frac{1}{m+1}.
\end{equation}
Thus we have the following:

\begin{corollary}
For time-independent differential equations, under the same assumptions of \thm{main}, there exists a quantum algorithm using
\begin{equation}
O\bigl(\kappa_Vs\|A\|Tq \, \poly(\log(\kappa_Vs\gamma\|A\|\|f\|T/\epsilon g))\bigr)
\end{equation}
queries to $O_A(h,l)$, $O_x$, and $O_f(h,l)$. The gate complexity of this algorithm is larger than its query complexity by a factor of $\poly(\log(\kappa_Vds\gamma\|A\|\|f\|T/\epsilon))$.
\end{corollary}

The complexity of our algorithm depends on the parameter $q$ in defined in \eq{main_thm_parameters}, which characterizes the decay of the final state relative to the initial state. As discussed in Section 8 of \cite{BCOW17}, it is unlikely that the dependence on $q$ can be significantly improved, since renormalization of the state effectively implements postselection and an efficient procedure for performing this would have the unlikely consequence $\mathrm{BQP} = \mathrm{PP}$.

We also require the real parts of the eigenvalues of $A(t)$ to be non-positive for all $t\in[0,T]$ so that the solution cannot grow exponentially. This requirement is essentially the same as in the time-independent case considered in \cite{BCOW17} and improves upon the analogous condition in \cite{Ber14} (which requires an additional stability condition). Also as in \cite{BCOW17}, our algorithm can produce approximate solutions for non-diagonalizable $A(t)$, although the dependence on $\epsilon$ degrades to $\poly(1/\epsilon)$. For further discussion of these considerations, see Sections 1 and 8 of \cite{BCOW17}.

\section{Boundary value problems}
\label{sec:BVP}

So far we have focused on initial value problems (IVPs). Boundary value problems (BVPs) are another widely studied class of differential equations that appear in many applications, but that can be harder to solve than IVPs.

Consider a sparse, linear, time-dependent system of differential equations as in \prb{ode} but with a constraint on some linear combination of the initial and final states:

\begin{problem}\label{prb:bvp}
In the \emph{quantum BVP}, we are given a system of equations
\begin{equation}
\frac{\d{x(t)}}{\d{t}}=A(t)x(t)+f(t),
\end{equation}
where $x(t) \in\C^d$, $A(t) \in \C^{d\times d}$ is $s$-sparse, and $f(t) \in \C^d$ for all $t\in[0,T]$, and a boundary condition $\alpha x(0)+\beta x(T)=\gamma$ with $\alpha, \beta, \gamma \in \C^d$. Suppose there exists a unique solution $\hat{x}\in C^{\infty}(0,T)$ of this boundary value problem. Given oracles that compute the locations and values of nonzero entries of $A(t)$ for any $t$, and that prepare quantum states $\alpha |x(0)\rangle+\beta |x(T)\rangle=|\gamma\rangle$ and $|f(t)\rangle$ for any $t$, the goal is to output a quantum state $|x(t^*)\rangle$ that is proportional to $x(t^*)$ for some specified $t^*\in[0,T]$.
\end{problem}

As before, we can rescale $[0,T]$ onto $[-1,1]$ by a linear mapping. However, since we have boundary conditions at $t=0$ and $t=T$, we cannot divide $[0,T]$ into small subintervals. Instead, we directly map $[0,T]$ onto $[-1,1]$ with a linear map $K$ satisfying $K(0)=1$ and $K(T)=-1$:
\begin{equation}
K\colon t \mapsto \overline{t}=1-\frac{2t}{T}.
\end{equation}
Now the new differential equations are
\begin{equation}
\frac{\d{x}}{\d\overline{t}}=-\frac{T}{2}\bigl(A(t)x+f(t)\bigr).
\end{equation}
If we define $A_K(\overline{t}):=-\frac{T}{2}A(t)$ and $f_K(\overline{t})=-\frac{T}{2}f(t)$, we have
\begin{equation}
\frac{\d{x}}{\d\overline{t}}=A_K(\overline{t})x(\overline{t})+f_K(\overline{t})
\end{equation}
for $\overline{t}\in[-1,1]$. Now the boundary condition takes the form
\begin{equation}
\alpha x(1)+\beta x(-1)=\gamma.
\end{equation}

Since we only have one solution interval, we need to choose a larger order $n$ of the Chebyshev series to reduce the solution error. In particular, we take
\begin{equation}
n=\frac{e}{2}\|A\|T\max\biggl\{\biggl\lfloor\frac{\log(\Omega)}{\log(\log(\Omega))}\biggr\rfloor,\biggl\lfloor\frac{\log(\omega)}{\log(\log(\omega))}\biggr\rfloor\biggr\}
\label{eq:bvp_n}
\end{equation}
where $\Omega$ and $\omega$ are the same as \thm{main}.

As in \sec{linear_system}, we approximate $x(t)$ by a finite Chebyshev series with interpolating nodes $\{t_l=\cos\frac{l\pi}{n} : l \in \range{n}\}$ and thereby obtain a linear system
\begin{equation}
\frac{\d{x(t_l)}}{\d{t}}=A_K(t_l)x(t_l)+f(t_l), \quad l\in\range{n}
\end{equation}
with the boundary condition
\begin{equation}
\alpha x(t_0)+\beta x(t_n)=\gamma.
\label{eq:boundary_condition}
\end{equation}

Observe that the linear equations have the same form as in \eq{interpolation_diff}. Instead of \eq{interpolation_init}, the term with $l=0$ encodes the condition \eq{boundary_condition} expanded in a Chebyshev series, namely
\begin{equation}
\alpha_i\sum_{k=0}^nc_{i,k}T_k(t_0)+\beta_i\sum_{k=0}^nc_{i,k}T_k(t_n)=\gamma_i
\end{equation}
for each $i\in\rangez{d}$.
Since $T_k(t_0)=1$ and $T_k(t_n)=(-1)^k$, this can be simplified as
\begin{equation}
\sum_{k=0}^n(\alpha_i+(-1)^k\beta_i)c_{i,k}=\gamma_i.
\end{equation}
If $\alpha_i+(-1)^k\beta_i=0$, the element of $|il\rangle\langle ik|$ of $L_2(A_K)$ is zero; if $\alpha_i+(-1)^k\beta_i\ne0$, without loss of generality, the two sides of this equality can be divided by $\alpha_i+(-1)^k\beta_i$ to guarantee that the terms with $l=0$ can be encoded as in \eq{interpolation_init}.

Now this system can be written in the form of equation \eq{linear_system} with $m=1$. Here
$L$, $|X\rangle$, and $|B\rangle$ are the same as in \eq{matrixL}, \eq{vectorX}, and \eq{vectorB}, respectively, with $m=1$, except for adjustments to $L_3$ that we now describe.

The matrix $L_3$ represents the linear combination $x_i(t^*)=\sum_{k=0}^nc_{i,k}T_k(t^*)$. Thus we take
\begin{equation}
L_3=\sum_{i=0}^{d-1}\sum_{k=0}^nT_k(t^*)|i0\rangle\langle ik|.
\end{equation}
Since $|T_k(t^*)|\le1$, we have
\begin{equation}
\| L_3\| \le n+1,
\end{equation}
and it follows that \lem{error} also holds for boundary value problems. Similarly, \lem{condition_number} still holds with $m=1$.

We are now ready to analyze the complexity of the quantum BVP algorithm. The matrix $L$ defined above is a $(p+2)d(n+1)\times(p+2)d(n+1)$ matrix with $O(ns)$ nonzero entries in any row or column, with condition number $O(\kappa_Vpn^{3.5})$. By \lem{success_prob} with $p=O(1)$, $O(q/\sqrt{n})$ repetitions suffice to ensure success probability $\Omega(1)$. By \eq{bvp_n}, $n$ is linear in $\|A\|T$ and poly-logarithmic in $\Omega$ and $\omega$. Therefore, we have the following:

\begin{theorem}\label{thm:bvp}
Consider an instance of the quantum BVP as defined in \prb{bvp}. Assume $A(t)$ can be diagonalized as $A(t)=V(t)\Lambda(t)V^{-1}(t)$ where $\Lambda(t)=\diag(\lambda_1(t),\ldots,\lambda_d(t))$ with $\Re(\lambda_i(t))\le0$ for each $i \in \rangez{d}$ and $t\in[0,T]$.
Then there exists a quantum algorithm that produces a state $x(t^*)/\| x(t^*)\|$ $\epsilon$-close to $\hat{x}(t^*)/\| \hat{x}(t^*)\|$ in $l^2$ norm, succeeding with probability $\Omega(1)$, with a flag indicating success, using
\begin{equation}
O\bigl(\kappa_Vs\|A\|^4T^4q \, \poly(\log(\kappa_Vs\|A\|g'T/\epsilon g))\bigr)
\end{equation}
queries to $O_A(h,l)$, $O_x$, and $O_f(h,l)$.
Here $\|A\|$, $\kappa_V$, $g$, $g'$ and $q$ are defined as in \thm{main}.
The gate complexity is larger than the query complexity by a factor of $\poly(\log(\kappa_Vds\|A\|g'T/\epsilon))$.
\end{theorem}

As for initial value problems, we can simplify this result in the time-independent case.

\begin{corollary}
For a time-independent boundary value problem, under the same assumptions of \thm{bvp}, there exists a quantum algorithm using

\begin{equation}
O\bigl(\kappa_Vs\|A\|^4T^4q \, \poly(\log(\kappa_Vs\gamma\|A\|\|f\|T/\epsilon g))\bigr)
\end{equation}
queries to $O_A(h,l)$, $O_x$, and $O_f(h,l)$. The gate complexity of this algorithm is larger than its query complexity by a factor of $\poly(\log(\kappa_Vds\gamma\|A\|\|f\|T/\epsilon))$.
\end{corollary}

\section{Discussion}
\label{sec:discussion}

In this paper, we presented a quantum algorithm to solve linear, time-dependent ordinary differential equations. Specifically, we showed how to employ a global approximation based on the spectral method as an alternative to the more straightforward finite difference method. Our algorithm handles time-independent differential equations with almost the same complexity as \cite{BCOW17}, but unlike that approach, can also handle time-dependent differential equations. Compared to \cite{Ber14}, our algorithm improves the complexity of solving time-dependent linear differential equations from $\poly(1/\epsilon)$ to $\poly(\log(1/\epsilon))$.

This work raises several natural open problems. First, our algorithm must assume that the solution is smooth. If the solution is in $C^r$, the solution error is $O(\frac{1}{n^{r-2}})$ by \lem{spectral_approx}. Can we improve the complexity to $\poly(\log(1/\epsilon))$ under such weaker smoothness assumptions?

Second, the complexity of our algorithm is logarithmic in the parameter $g'$ defined in \eq{main_thm_parameters}, which characterizes the amount of fluctuation in the solution. However, the query complexity of Hamiltonian simulation is independent of that parameter \cite{PQSV11,BCC13}. Can we develop quantum algorithms for general differential equations with query complexity independent of $g'$?

Third, our algorithm has nearly optimal dependence on $T$, scaling as $O(T \poly(\log T))$. According to the no-fast-forwarding theorem \cite{BAC07}, the complexity must be at least linear in $T$, and indeed linear complexity is achievable for the case of Hamiltonian simulation \cite{Chi08}. Can we handle general differential equations with complexity linear in $T$? Furthermore, can we achieve an optimal tradeoff between $T$ and $\epsilon$ as shown for Hamiltonian simulation in \cite{LC17}?

Finally, can the techniques developed here be applied to give improved quantum algorithms for linear partial differential equations, or even for nonlinear ODEs or PDEs?

\section*{Acknowledgements}

We thank Stephen Jordan and Aaron Ostrander for valuable discussions of quantum algorithms for time-dependent linear differential equations.

This work was supported in part by the Army Research Office (MURI award W911NF-16-1-0349), the Canadian Institute for Advanced Research, the National Science Foundation (grants 1526380 and 1813814), and the U.S.\ Department of Energy, Office of Science, Office of Advanced Scientific Computing Research, Quantum Algorithms Teams and Quantum Testbed Pathfinder programs.


\appendix

\section{Chebyshev polynomials}
\label{app:Chebyshev}

This appendix defines the Chebyshev polynomials and presents some of their properties that are useful for our analysis.

For any $k \in \N$, the Chebyshev polynomial of the first kind can be defined as the function
\begin{equation}
T_k(x)=\cos(k\arccos x),\quad x\in[-1,1].
\end{equation}
It can be shown that this is a polynomial of degree $k$ in $x$. For example, we have
\begin{align}
  T_0(x)&=1, & T_1(x)&=x, & T_2(x)&=2x^2-1, & T_3(x)&=4x^3-3x, & T_4(x)&=8x^4-8x^2+1.
\end{align}

Using the trigonometric addition formula $\cos(k+1)\theta+\cos(k-1)\theta=2\cos\theta\cos k\theta$, we have the recurrence
\begin{equation}
T_{k+1}(x)=2xT_k(x)-T_{k-1}(x)
\end{equation}
(which also provides an alternative definition of the Chebyshev polynomials, starting from the initial conditions $T_0(x)=1$ and $T_1(x)=x$).
We also have the bounds
\begin{equation}
\begin{aligned}
|T_k(x)|&\le1 \text{~for~} |x|\le1, & T_k(\pm1)&=(\pm1)^k.
\end{aligned}
\end{equation}

Chebyshev polynomials are orthogonal polynomials on $[-1,1]$ with the weight function $w(x):=(1-x^2)^{-1/2}$. More concretely, defining an inner product on $L_w^2(-1,1)$ by
\begin{align}
(f,g)_w
&:=\int_{-1}^1 f(x)g(x)\frac{\d{x}}{\sqrt{1-x^2}},
\end{align}
we have
\begin{align}
(T_m,T_n)_w
&=\int_0^{\pi}\cos m\theta\cos n\theta \, \d{\theta} \\
&=\frac{\pi}{2}\sigma_n\delta_{m,n}
\label{eq:cheby_orth}
\end{align}
where
\begin{align}
  \sigma_n:=
  \begin{cases}
  2 & n=0\\
  1 & n\ge1.
\end{cases}
\label{eq:sigmadef}
\end{align}

It is well known from the approximation theorem of Weierstrass that $\{T_k(x) : k \in \N\}$ is complete on the space $L_w^2(-1,1)$. In other words, we have the following:
\begin{lemma}
Any function $u\in L_w^2(-1,1)$ can be expanded by a unique Chebyshev series as
\begin{equation}
u(x)=\sum_{k=0}^{\infty}\hat{c}_kT_k(x)
\end{equation}
where the coefficients are
\begin{equation}
\hat{c}_k=\frac{2}{\pi}(u,T_k)_w.
\end{equation}
\end{lemma}

For any $N \in \N$, we introduce the orthogonal projection $P_N \colon L_w^2(-1,1)\rightarrow\mathbb{P}_N$ (where $\mathbb{P}_N$ denotes the set of polynomials of degree at most $N$) by
\begin{equation}
P_N u(x)=\sum_{k=0}^N\hat{c}_kT_k(x).
\end{equation}
By the completeness of the Chebyshev polynomials, we have
\begin{equation}
(P_Nu(x),v(x))_w=(u(x),v(x))_w\quad\forall\, v\in\mathbb{P}_N
\end{equation}
and
\begin{equation}
\|P_Nu(x)-u(x)\|_w\rightarrow0,\quad N\rightarrow\infty.
\end{equation}

Finally, we compute the Chebyshev series of $u'(x)$ in terms of the Chebyshev series of $u(x)$. Since $T_k(x)=\cos k\theta$ where $\theta=\arccos x$, we have
\begin{equation}
T'_k(x)=\frac{k\sin k\theta}{\sin\theta}.
\end{equation}
Since
\begin{equation}
2\cos k\theta=\frac{\sin(k+1)\theta}{\sin\theta}-\frac{\sin(k-1)\theta}{\sin\theta},
\end{equation}
we obtain
\begin{equation}
2T_k(x)=\frac{T'_{k+1}(x)}{k+1}-\frac{T'_{k-1}(x)}{k-1},\quad k\ge2
\end{equation}
and
\begin{equation}
T_1(x)=\frac{T'_2(x)}{4}.
\end{equation}

Since $P_N u(x) \in \mathbb{P}_N$, the derivative of this projection should be in $\mathbb{P}_{N-1}$. Indeed, we have
\begin{equation}
\begin{aligned}
u'(x)
&=\sum_{k=0}^{N-1}\hat{c}'_kT_k(x) \\
&=\frac{1}{2}\sum_{k=1}^{N-1}\hat{c}'_k\frac{T'_{k+1}(x)}{k+1}
 -\frac{1}{2}\sum_{k=2}^{N-1}\hat{c}'_k\frac{T'_{k-1}(x)}{k-1}
 +\hat{c}'_0T_0(x)\\
&=\sum_{k=2}^{N-2}(\hat{c}'_{k-1}-\hat{c}'_{k+1})\frac{T'_k(x)}{2k}-\frac{1}{2}\hat{c}'_{2}T'_1(x)+\frac{1}{2}\hat{c}'_{N-2}\frac{T'_{N-1}(x)}{N-1}+\frac{1}{2}\hat{c}'_{N-1}\frac{T'_n(x)}{N}+\hat{c}'_0T_0(x)\\
&=\sum_{k=1}^N\hat{c}_kT'_k(x).
\end{aligned}
\end{equation}
Comparing the coefficients of both sides, we find
\begin{equation}
\begin{aligned}
\sigma_k\hat{c}'_k&=\hat{c}'_{k+2}+2(k+1)\hat{c}_{k+1},\quad k\in\rangez{n}\\
\hat{c}'_N&=0 \\
\hat{c}'_{N+1}&=0
\end{aligned}
\end{equation}
where $\sigma_k$ is defined in \eq{sigmadef}.

Since $\hat{c}'_k=0$ for $k\ge N$, we can calculate $\hat{c}'_{N-1}$ from $\hat{c}_N$ and then successively calculate $\hat{c}'_{N-2},\ldots,\hat{c}'_1,\hat{c}'_0$. This recurrence gives
\begin{equation}
\hat{c}'_k=\frac{2}{\sigma_k}\sum_{\substack{j=k+1 \\ \text{$j+k$ odd}}}^Nj\hat{c}_j,\quad  k\in\rangez{n}.
\end{equation}
Since $\hat{c}'_k$ only depends on $\hat{c}_j$ for $j>k$, the transformation matrix $D_N$ between the values $\hat{c}'_k$ and $\hat{c}_k$ for $k \in \rangez{n+1}$ is an upper triangular matrix with all zero diagonal elements, namely
\begin{equation}
[D_N]_{kj}=
\begin{cases}
\frac{2j}{\sigma_k} &j>k,~\text{$j+k$ odd} \\
0 & \text{otherwise}.
\end{cases}
\end{equation}

\section{An example of the quantum spectral method}
\label{app:example}

\sec{linear_system} defines a linear system that implements the quantum spectral method for solving a system of $d$ time-dependent differential equations. Here we present a simple example of this system for the case $d=1$, namely
\begin{equation}
\frac{\d{x}}{\d{t}}=A(t)x(t)+f(t)
\end{equation}
where $x(t),A(t),f(t)\in\C$, $t\in[0,T]$, and we have the initial condition
\begin{equation}
x(0)=\gamma \in \C.
\end{equation}

In particular, we choose $m=3$, $n=2$, and $p=1$ in the specification of the linear system. We divide $[0,T]$ into $m=3$ intervals $[0,\Gamma_1],[\Gamma_1,\Gamma_2],[\Gamma_2,T]$ with $\Gamma_0=0, \Gamma_m=T$, and map each one onto $[-1,1]$ with the linear mapping $K_h$ satisfying $K_h(\Gamma_h)=1$ and $K_h(\Gamma_{h+1})=-1$. Then we take the finite Chebyshev series of $x(t)$ with $n=2$ into the differential equation with interpolating nodes $\{t_l=\cos\frac{l\pi}{n} : l \in \range{2}\} = \{0,-1\}$ to obtain a linear system. Finally, we repeat the final state $p=1$ time to increase the success probability.

With these choices, the linear system has the form
\begin{equation}
L=
  \begin{pmatrix}
    L_1+L_2(A_0) &  &  &  &  \\
    L_3 & L_1+L_2(A_1) &  &  &  \\
     & L_3 & L_1+L_2(A_2) &  &  \\
     &  & L_3 & L_4 &  \\
     &  &  & L_5 & L_4 \\
  \end{pmatrix}
\end{equation}
with
\begin{align}
L_1&=|0\rangle\langle 0|P_n+\sum_{l=1}^n|l\rangle\langle l|P_nD_n=
  \begin{pmatrix}
    1 & 1 & 1 \\
    0 & 1 & 0 \\
    0 & 1 & -4 \\
  \end{pmatrix} \\
L_2(A_h)&=-\sum_{l=1}^nA_h(t_l)\otimes|l\rangle\langle l|P_n=-
  \begin{pmatrix}
    0 & 0 & 0 \\
    A_h(0) & 0 & -A_h(0) \\
    A_h(-1) & -A_h(-1) & A_h(-1) \\
  \end{pmatrix} \\
L_3&=\sum_{i=0}^{d}\sum_{k=0}^n(-1)^k|i0\rangle\langle ik|=
  \begin{pmatrix}
    1 & -1 & 1 \\
    0 & 0 & 0 \\
    0 & 0 & 0 \\
  \end{pmatrix} \\
L_4&=-\sum_{i=0}^{d}\sum_{l=1}^n|il\rangle\langle il-1|+\sum_{i=0}^{d}\sum_{l=0}^n|il\rangle\langle il|=
  \begin{pmatrix}
    1 & 0 & 0 \\
    -1 & 1 & 0 \\
    0 & -1 & 1 \\
  \end{pmatrix} \\
L_5&=-\sum_{i=0}^{d}|i0\rangle\langle in|=
  \begin{pmatrix}
    0 & 0 & -1 \\
    0 & 0 & 0 \\
    0 & 0 & 0 \\
  \end{pmatrix}.
\end{align}

The vector $|X\rangle$ has the form
\begin{equation}
|X\rangle=
  \begin{pmatrix}
    c_0(\Gamma_1) \\
    c_1(\Gamma_1) \\
    c_2(\Gamma_1) \\
    c_0(\Gamma_2) \\
    c_1(\Gamma_2) \\
    c_2(\Gamma_2) \\
    c_0(\Gamma_3) \\
    c_1(\Gamma_3) \\
    c_2(\Gamma_3) \\
    x \\
    x \\
    x \\
    x \\
    x \\
    x \\
  \end{pmatrix}
\end{equation}
where $c_l(\Gamma_{h+1})$ are the Chebyshev series coefficients of $x(\Gamma_{h+1})$ and $x$ is the final state $x(\Gamma_m)=x(-1)$.

Finally, the vector $|B\rangle$ has the form
\begin{equation}
|B\rangle=
  \begin{pmatrix}
    \gamma \\
    f_0(0) \\
    f_0(-1) \\
    0 \\
    f_1(0) \\
    f_1(-1) \\
    0 \\
    f_2(0) \\
    f_2(-1) \\
    0 \\
    0 \\
    0 \\
    0 \\
    0 \\
    0 \\
  \end{pmatrix}
\end{equation}
where $\gamma$ comes from the initial condition and $f_{h}(\cos\frac{l\pi}{n})$ is the value of $f_h$ at the interpolation point $t_l=\cos\frac{l\pi}{n} \in \{0,-1\}$.

\end{document}